\documentclass[11pt]{article}

\pdfoutput=1

\usepackage{a4}
\usepackage[top=30truemm,bottom=30truemm]{geometry}

\usepackage[running]{lineno}

\usepackage{authblk}

\usepackage{graphicx}
\usepackage{mathrsfs}
\usepackage{amsmath,amssymb}
\usepackage{amsthm}
\usepackage{algorithmic}
\usepackage{algorithm}
\usepackage{wrapfig}

\newtheorem{theorem}{Theorem}
\newtheorem{lemma}{Lemma}

\newtheorem{example}{Example}

\newtheorem{observation}{Observation}

\theoremstyle{definition}

\newcommand{\MAW}{\mathsf{MAW}}
\newcommand{\rle}{\mathsf{rle}}

\newcommand{\bridge}{bridge}

\newcommand{\tRLESA}{\mathsf{tRLESA}}

\usepackage[final,mode=multiuser]{fixme}
\fxsetup{theme=color}
\FXRegisterAuthor{ta}{ata}{TA}
\FXRegisterAuthor{ko}{ako}{KO}
\FXRegisterAuthor{tm}{atm}{TM}
\FXRegisterAuthor{yn}{ayn}{YN}
\FXRegisterAuthor{si}{asi}{SI}
\fxsetface{env}{}
\newtheorem{proposition}{Proposition}

\theoremstyle{definition}

\newcommand{\shrink}{\mathsf{shk}}
\newcommand{\Exp}{\mathsf{Exp}}
\newcommand{\Expplus}{\mathsf{Exp}_{\texttt{+}}}

\title{Minimal Absent Words on Run-Length Encoded Strings}

\author{Tooru~Akagi$^{1}$}
\author{Kouta~Okabe$^{2}$}
\author{Takuya~Mieno$^{3}$}
\author{Yuto~Nakashima$^{1}$}
\author{Shunsuke~Inenaga$^{1,4}$}

\affil{
  \normalsize{
  \textit{$^1$Department of Informatics, Kyushu University, Japan}\\
  \texttt{\{toru.akagi, yuto.nakashima, inenaga\}@inf.kyushu-u.ac.jp}\\
  \textit{$^2$Department of Information Science and Technology, Kyushu University, Japan}\\
  \texttt{okabe.kouta.966@s.kyushu-u.ac.jp}\\
  \textit{$^3$Faculty of Information Science and Technology, Hokkaido University, Japan\footnote{Current affiliation: University of Electro-Communications, Japan~(\texttt{tmieno@uec.ac.jp})}} \\
  \texttt{takuya.mieno@ist.hokudai.ac.jp}\\
  \textit{$^4$PRESTO, Japan Science and Technology Agency, Japan}
  }
}

\date{}

\begin{document}
\maketitle

\begin{abstract}
  A string $w$ is called a \emph{minimal absent word} for another string $T$
  if $w$ does not occur (as a substring) in $T$ and any proper substring of $w$ occurs in $T$.
  State-of-the-art data structures for reporting the set $\MAW(T)$ of MAWs from a given string $T$ of length $n$ require $O(n)$ space, can be built in $O(n)$ time, and can report all MAWs in $O(|\MAW(T)|)$ time upon \tmnote*{added}{a} query.
  This paper initiates the problem of computing MAWs from \tmnote*{added}{a} compressed representation of a string.
  In particular, we focus on the most basic compressed representation of a string,
  \emph{run-length encoding} (\emph{RLE}), which represents
  each maximal run of the same characters $a$ by $a^p$ where $p$ is the length of the run.
  Let $m$ be the RLE-size of string $T$.
  After categorizing the MAWs into five disjoint sets $\mathcal{M}_1$, $\mathcal{M}_2$, $\mathcal{M}_3$, $\mathcal{M}_4$, $\mathcal{M}_5$ using RLE,
  we present matching upper and lower bounds for the number of MAWs in $\mathcal{M}_i$ for $i = 1,2,4,5$ in terms of RLE-size $m$, except for $\mathcal{M}_3$ whose size is unbounded by $m$.
  We then present a compact $O(m)$-space data structure that can report
  all MAWs in optimal $O(|\MAW(T)|)$ time.
\end{abstract}

\section{Introduction}

An \emph{absent word} (a.k.a. \emph{a forbidden word}) for a string $T$
is a non-empty string that is \emph{not} a substring of $T$.
An absent word $X$ for $T$ is said to be a \emph{minimal absent word} (\emph{MAW})
for $T$ if all proper substrings of $X$ occur in $T$.
MAWs are combinatorial string objects, and their interesting
mathematical properties have extensively been studied in the literature (see \cite{Beal1996MAWandSymbolicDynamics,Crochemore1998MAWdefinition,Fici2006MAWapplication,CrochemoreHKMPR20,MienoKAFNIBT20,abs-2105-08496} and references therein).
MAWs also enjoy several applications including
phylogeny~\cite{Chairungsee2012PhylogenyByMAW},
data compression~\cite{Crochemore2000DCA,crochemore2002improved,AyadBFHP21},
musical information retrieval~\cite{CrawfordB018},
and bioinformatics~\cite{Almirantis2017MolecularBiology,Charalampopoulos18,pratas2020persistent,koulouras2021significant}.

Thus, given a string $T$ of length $n$ over an alphabet of size $\sigma$,
computing the set $\MAW(T)$ of all MAWs for $T$ is an interesting and important problem:
Crochemore et al.~\cite{Crochemore1998MAWdefinition} presented
the first efficient data structure of $O(n)$ space
which outputs all MAWs in $\MAW(T)$ in $O(\sigma n)$ time and $O(n)$ working space.
Since the number $|\MAW(T)|$ of MAWs for $T$ can be as large as $O(\sigma n)$
and there exist strings $S$ for which $|\MAW(S)| \in \Omega(\sigma |S|)$~\cite{Crochemore1998MAWdefinition},
Crochemore et al.'s algorithm~\cite{Crochemore1998MAWdefinition} runs in optimal time in the worst case.
Later, Fujishige et al.~\cite{Fujishige2016DAWG} presented an improved
data structure of $O(n)$ space,
which can report all MAWs in $O(n+|\MAW(T)|)$ time and $O(n)$ working space.
Fujishige et al.'s algorithm~\cite{Fujishige2016DAWG} can easily be modified so
it uses $O(|\MAW(T)|)$ time for reporting all MAWs,
by explicitly storing all MAWs when $|\MAW(T)| \in O(n)$.
The key tool used in these two algorithms is an $O(n)$-size automaton
called the \emph{DAWG}~\cite{BlumerBHECS85},
which accepts all substrings of $T$.
The DAWG for string $T$ can be built in $O(n \log \sigma)$ time for general ordered alphabets~\cite{BlumerBHECS85}, or in $O(n)$ time for integer alphabets of size polynomial in $n$~\cite{Fujishige2016DAWG}.
There also exist other efficient algorithms for computing MAWs
with other string data structures such as suffix arrays and Burrows-Wheeler transforms~\cite{Belazzougui2013ESA,Barton2014MAWbySA}.
MAWs in other settings have also been studied in the literature,
including length specified versions~\cite{charalampopoulos2018extended},
the sliding window versions~\cite{CrochemoreHKMPR20,MienoKAFNIBT20,abs-2105-08496},
circular string versions~\cite{FiciRR19},
and labeled tree versions~\cite{FiciG19}.

In this paper, we initiate the study of computing MAWs for \emph{compressed} strings.
As the first step of this line of research,
we consider strings which are compactly represented by \emph{run-length encoding} (\emph{RLE}).
Let $m$ be the size of the RLE of an input string $T$.
We first categorize the elements of $\MAW(T)$ into five disjoint subsets
$\mathcal{M}_1$, $\mathcal{M}_2$, $\mathcal{M}_3$, $\mathcal{M}_4$, and $\mathcal{M}_5$, by considering how the MAWs can be related to the boundaries of maximal character runs in $T$ (Section~\ref{sec:prelim}).
In Section~\ref{sec:upper} and Section~\ref{sec:lower},
we present matching upper bounds and lower bounds for their sizes $|\mathcal{M}_i|$~($i = 1, 2, 4, 5$) in terms of the RLE size $m$ or the number $\sigma'_T$ of distinct characters occurring in $T$.
Notice that $\sigma'_T \leq m$ always holds.
The exception is $\mathcal{M}_3$, which can contain $\Omega(n)$ MAWs regardless of the RLE size $m$.
Still, in Section~\ref{sec:representation}
we propose our RLE-compressed $O(m)$-space data structure that can enumerate all MAWs for $T$ in output-sensitive $O(|\MAW(T)|)$ time.
Since $m \leq n$ always holds,
our result is an improvement over 
Crochemore et al.'s and Fujishige et al.'s results both of which require $O(n)$ space to store representations of all MAWs.
Charalampopoulos et al.~\cite{charalampopoulos2018extended} showed how one can use \emph{extended bispecial factors} of $T$ to represent all MAWs for $T$ in $O(n)$ space, and to output all MAWs in optimal $O(|\MAW(T)|)$ time upon a query.
While the way how we characterize the MAWs 
may be seen as the RLE version of their method based on the extended bispecial factors,
our $O(m)$-space data structure cannot be obtained by a straightforward extension from~\cite{charalampopoulos2018extended}, 
since there exists a family of strings over a constant-size alphabet for which
the RLE-size is $m \in O(1)$ but $|\MAW(T)| \in \Omega(n)$.
We note that, by the use of \emph{truncated RLE suffix arrays}~\cite{TamakoshiGIBT15},
our $O(m)$-space data structure can be built in $O(m \log m)$ time with $O(m)$ working space (the details of the construction will be presented in the full version of this paper).

\section{Preliminaries}
\label{sec:prelim}

\subsection{Strings}
Let $\Sigma$ be an ordered alphabet.
An element of $\Sigma$ is called a character.
An element of $\Sigma^*$ is called a string.
The length of a string $T$ is denoted by $|T|$.
The empty string $\varepsilon$ is the string of length 0.
If $T = xyz$, then $x$, $y$, and $z$ are called
a \emph{prefix}, \emph{substring}, and \emph{suffix} of $T$, respectively.
They are called a \emph{proper prefix}, \emph{proper substring},
and \emph{proper suffix} of $T$ if $x \neq T$, $y \neq T$, and $z \neq T$,
respectively.
For any $1 \le i \le |T|$, the $i$-th character of $T$ is denoted by $T[i]$.
For any $1 \le i \le j \le |T|$, $T[i..j]$ denotes
the substring of $T$ starting at $i$ and ending at $j$.
For any $i \le |T|$ and $1 \le j$, let $T[..i] = T[1..i]$ and $T[j..] = T[j..|T|]$.
We say that a string $w$ \emph{occurs} in a string $T$
if $w$ is a substring of $T$.
Note that by definition, the empty string $\varepsilon$
is a substring of any string $T$
and hence $\varepsilon$ always occurs in $T$.

\sinote*{modified}{%
Let $\#_Tw$ denote the number of occurrences of a string $w$ in a string $T$.
We will abbreviate it to $\# w$ when no confusion occurs.
}%

\subsection{Run length encoding (RLE) and bridges}

The {\em run-length encoding} $\rle(T)$ of string $T$ is
a compact representation of $T$ such that
each maximal run of the same characters in $T$ is represented
by a pair of the character and the length of the maximal run.
More formally, $\rle(T)= a_1^{p_1} \cdots a_m^{p_m}$ encodes each substring $T[i.. i+p-1]$ by $a^p$ if
$T[j] = a \in \Sigma$ for every $i \le j \le i+p-1$, $T[i-1] \ne T[i]$, and $T[i+p-1] \ne T[i+p]$.
Each $a^p$ in $\rle(T)$ is called a (character) \emph{run},
and $p$ is called the exponent of this run.
The $j$-th maximal run in $\rle(T)$ is denoted by $r_j$,
namely $\rle(T) = r_1 \cdots r_m$.
The \emph{size} of $\rle(T)$, denoted $R(T)$, is the number of maximal character runs in $\rle(T)$.
E.g., for a string $T = \mathtt{aacccccccbbabbbb}$ of length 18,
$\rle(T) = \mathtt{a}^2 \mathtt{c}^7 \mathtt{b}^2 \mathtt{a}^1 \mathtt{b}^4$ and $R(T) = 5$.

  Our model of computation is a standard word RAM with machine word size $\Omega(\log |T|)$,
  and the space requirements of our data structures will be measured by the number of words (not bits).
  Thus, $\rle(T)$ of size $m$ can be stored in $O(m)$ space.

\subsection{Bridges}

A string $w \in \Sigma^*$ of length $|w| \geq 2$ is said to be a \emph{\bridge}\ if
$w[1] \neq w[2]$ and $w[|w|-1] \neq w[|w|]$.
In other words, both of the first run and the last run in $\rle(w)$ are of length $1$.
A substring of $T$ that is a \bridge\ is called a \bridge\ substring of $T$.
Let $B_\ell$ denote the set of bridge substrings $w$ of $T$ with $R(w) = \ell$.
Further let $\mathcal{B} = \bigcup_{\ell}B_\ell$ be the set of all bridge substrings of $T$.
For example, for the same string $T = \mathtt{aacccccccbbabbbb}$ as the above one,
the substring $\mathtt{ac}^7\mathtt{b}^2\mathtt{a}$ of $T$ is a bridge, and
$B_4 = \{
\mathtt{ac}^7\mathtt{b}^2\mathtt{a},
\mathtt{cb}^2\mathtt{a}^1\mathtt{b}
\}$.
For a string $w$ with $R(w) \ge 3$,
we can obtain a bridge substring of $w$
by removing the first and the last runs of $w$
and then \emph{shrinking} the runs at both ends so that their exponents are $1$.
We denote by $\shrink(w)$ such shrunk bridge.
For convenience, let $\shrink(w) = \varepsilon$ if $R(w) \le 2$.
Also, for every $k \ge 2$, we denote $\shrink^k(w) = \shrink(\shrink^{k-1}(w))$.
For example, consider the same $T$ as the above again,
$\shrink(T) = \mathtt{acccccccbbab}$,
$\shrink^2(w) = \mathtt{cbba}$,
$\shrink^3(w) = \mathtt{b}$, and
$\shrink^k(w) = \varepsilon$ for any $k \ge 4$.

\subsection{Minimal absent words (MAWs)}

A string $w \in \Sigma^*$ is called an \emph{absent word} for a string $T$
if $w$ does not occur in $T$, namely if $\#w = 0$.
An absent word $w$ for $T$ is called a \emph{minimal absent word}
or \emph{MAW} for $T$ if
all proper substrings of $w$ occur in $T$.
We denote by $\MAW(T)$ the set of all MAWs for $T$.
An alternative definition of MAWs is such that
a string $aub$ of length at least two with $a, b \in \Sigma$ and $u \in \Sigma^*$ is a MAW of $T$ if $\#(aub) = 0$, $\#(au) \geq 1$ and $\#(ub) \geq 1$.
For a MAW of length $1$ (namely a character not occurring in $T$),
we use a convention that $u = \varepsilon$ and $a$ and $b$ are united into a single character.

The MAWs in $\MAW(T)$ are partitioned into the following five disjoint subsets $\mathcal{M}_i$~($1 \leq i \leq 5$) based on their RLE sizes $R(aub)$:
\begin{itemize}
 \item $\mathcal{M}_1 = \{aub \in \MAW(T) \mid R(aub) = 1\}$;
 \item $\mathcal{M}_2 = \{aub \in \MAW(T) \mid R(aub) = 2, u = \varepsilon\}$;
 \item $\mathcal{M}_3 = \{aub \in \MAW(T) \mid R(aub) = 3, a \neq u[1] \mbox{ and } b \neq u[ |u|]\}$;
 \item $\mathcal{M}_4 = \{aub \in \MAW(T) \mid R(aub) \geq 4, a \neq u[1] \mbox{ and } b \neq u[|u|]\}$;
 \item $\mathcal{M}_5 = \{aub \in \MAW(T) \mid R(aub) \geq 2, a = u[1] \mbox{ or } b = u[|u|]\}$.  
\end{itemize}
For $1 \leq i \leq 5$, a MAW $aub$ in $\mathcal{M}_i$ is called of \emph{type $i$}.

\sinote*{modified}{%
In the rest of this paper, we will consider an arbitrarily fixed string $T$ of length $n$.
For convenience, we assume that $n \ge 3$ and
that there are special terminal symbols $T[1] = T[n] = \$ \not\in\Sigma$
not occurring inside $T$.
Since $\$ \notin \Sigma$,
we do not consider any MAW containing $\$$ for $T$ in our arguments to follow
(recall that a MAW must be an element of $\Sigma^*$).
In addition, since $\$$ does not occur elsewhere in $T$,
$\MAW(T) = \MAW(T[2..n-1])$ holds.
}%

\begin{example}
Consider $T = \mathtt{\$}\mathtt{b}^2\mathtt{ac}^3\mathtt{ba}^2\mathtt{\$} = \mathtt{\$bbacccbaa\$}$. All MAWs in $\MAW(T)$ are divided into the following five types:
$\mathcal{M}_1 = \{\mathtt{aaa,bbb,cccc}\}$;
$\mathcal{M}_2 = \{\mathtt{ca,bc}\}$;
$\mathcal{M}_3 = \{\mathtt{acb,accb}\}$;
$\mathcal{M}_4 = \{\mathtt{cbac}\}$;
$\mathcal{M}_5 = \{\mathtt{bbaa}\}$.
\end{example}

Let $\Sigma'$ denote the set of characters occurring in $T$ except for $\$$.
Let $\sigma' = |\Sigma'|$ be the number of distinct characters occurring in $T[2.. n-1]$.

\section{Upper bounds on the number of MAWs for RLE strings}
\label{sec:upper}

In this section, we present upper bounds for the number of MAWs in a string $T$ that is represented by its RLE $\rle(T)$ of size $R(T) = m$.

\subsection{Upper bounds for the number of MAWs of type 1, 2, 3, 5}

We first consider the number of MAWs except for those of type 4.

\begin{lemma}
\label{M1_number}
$| \mathcal{M}_1 | = \sigma$.
\end{lemma}
\begin{proof}
By the definition of $\mathcal{M}_1$, any MAW in $\mathcal{M}_1$ is of the form $a^k$.
For any character $\alpha \in \Sigma'$ that occurs in $T$,
let $aub = \alpha^{p+1}$ such that $\alpha^p$ is the \emph{longest} maximal run of $\alpha$ in $T$.
Clearly $\alpha^{p} = au = ub$ occurs in $T$ and $\alpha^{p+1}$ does not occur in $T$. Since $R(aub) =R(\alpha^{p+1}) = 1$, $\alpha^{p+1} \in \mathcal{M}_1$ and it is the unique MAW of type 1 consisting of $\alpha$'s. 
For any character $\beta \in \Sigma \setminus \Sigma'$ that does not occur in $T$, clearly $\beta$ is a MAW of $T$ and $\beta \in \mathcal{M}_1$ since $R(\beta) = 1$.
In total, we obtain $|\mathcal{M}_1| = \sigma$.
\end{proof}
Note that this upper bound for $|\mathcal{M}_1|$ is tight for any string $T$ and alphabet $\Sigma$ of size $\sigma$.

\begin{lemma}
\label{M2_number}
$|\mathcal{M}_2| \in O((\sigma')^2)$.
\end{lemma}
\begin{proof}
Any MAW in $\mathcal{M}_2$ is of the form $ab$ with $a, b \in \Sigma$ and $a \ne b$.
By the definition of MAWs, $ab$ can be a MAW for $T$ only if both $a$ and $b$ occur in $T$, which implies that $a, b \in \Sigma'$.
The number of such combinations of $a$ and $b$ is $\sigma'(\sigma'-1)$.
\end{proof}
Since $\sigma' \leq m$ always holds, we have that $| \mathcal{M}_2 | \in O(m^2)$.
Later we will show that this upper bound for $|\mathcal{M}_2|$ is asymptotically tight.

\begin{lemma}
\label{M3_number}
$| \mathcal{M}_3 |$ is unbounded by $m$.
\end{lemma}
\begin{proof}
Consider a string $T = ac^{n-2}b$, where $a \neq c$ and $c \neq b$. 
Then $ac^kb$ for each $1 \leq k \leq n-3$ is a MAW of $T$ and $R(ac^kb) = 3$.
Since they are the only type 3 MAWs of $T$, we have that $|\mathcal{M}_3| = n-3$.
Clearly, the original length $n$ of $T$ cannot be bounded by $m = R(T) = 3$.   
\end{proof}
Although the number of MAWs of type 3 is unbounded by $m$, later we will present an $O(m)$-space data structure that can enumerate all elements in $\mathcal{M}_3$ in output-sensitive time.

\begin{lemma}
\label{M5_number}
$| \mathcal{M}_5 | \in O(m)$.
\end{lemma}

\begin{proof}
\sinote*{modified}{%
  Any MAW $aub \in \mathcal{M}_5$ can be represented by $a^{i+1}vb$ or $avb^{i+1}$ with maximal integer $i \geq 1$, where $a^{i}v = u$ in the former and $vb^{i} = u$ in the latter.
Let us consider the case of $a^{i+1}vb$ as the case of $avb^{i+1}$ is symmetric.
Then $ca^{i}vb$ with some character $c \neq a$ must occur in $T$.
Let $k$ be the beginning position of an occurrence of $ca^{i}vb$ in $T$. Then, $T[k+1..k+i] = a^{i}$ is a maximal run of $a$.

Now consider any distinct MAW $a^{i+1} v'b' \in \mathcal{M}_5 \setminus \{a^{i+1} vb\}$ with $v'b' \neq vb$.
Again, $c'a^{i}v'b'$ with some character $c' \neq a$ must occur in $T$.
Suppose on the contrary that $c'a^{i}v'b'$ has an occurrence beginning at the same position $k$ as $ca^{i} vb$.
This implies that $c' = c$, and both $a^{i}vb$ and $a^{i}v'b'$ are prefixes of $T[k+1..|T|]$.
\begin{itemize}
\item If $|a^ivb| < |a^iv'b'|$, then $a^iv'$ contains $a^ivb$ as a substring. Since $a^{i+1}v'$ occurs in $T$, $a^{i+1}vb$ must also occur in $T$. Hence $a^{i+1}vb$ is not a MAW for $T$, a contradiction.

\item If $|a^ivb| > |a^iv'b'|$, then $a^iv$ contains $a^iv'b'$ as a substring. Thus $a^{i+1}vb$ is an absent word for $T$ but it is not minimal. Hence $a^{i+1}vb$ is not a MAW for $T$, a contradiction.

\item If $|a^ivb| = |a^iv'b'|$, then this contradicts that $a^iub \neq a^iu'b'$.
\end{itemize}
Hence, at most two element of $\mathcal{M}_5$ can be associated
with a position $k$ in $T$ such that $T[k] \neq T[k+1]$.
The number of such positions does not exceed $2m$. 
}%
\end{proof}

\subsection{Upper bound for the number of MAWs of type 4} \label{subsec:upperbound_M4}
In the rest of this section, we show an upper bound of the number of MAWs of type 4.
Namely, we prove the following lemma.

\begin{lemma}
\label{M4_number}
$| \mathcal{M}_4 | \in O(m^2)$.
\end{lemma}

Firstly, we explain a way to characterize MAWs of type 4.
For any string $w \in \Sigma^*$ and integer $t > 0$, let $\Exp^t(w)$ be the set of bridges such that
$\Exp^t(w) = \{ w' \in \mathcal{B} \mid \shrink^{t}(w') = w \}$.
Namely, $\Exp^t(w)$ is the \emph{inverse image} of $\shrink^t(w') = w$ for bridge substrings $w'$ of $T$.
We use $\Exp(w)$ to denote $\Exp^1(w)$. 
Figure~\ref{fig:setK} gives an example for $\Exp^t(w)$ ($\Expplus^{t}(w)$ in the figure will be defined later).
\begin{figure}[tbh]
\begin{center}
\includegraphics[scale=0.5]{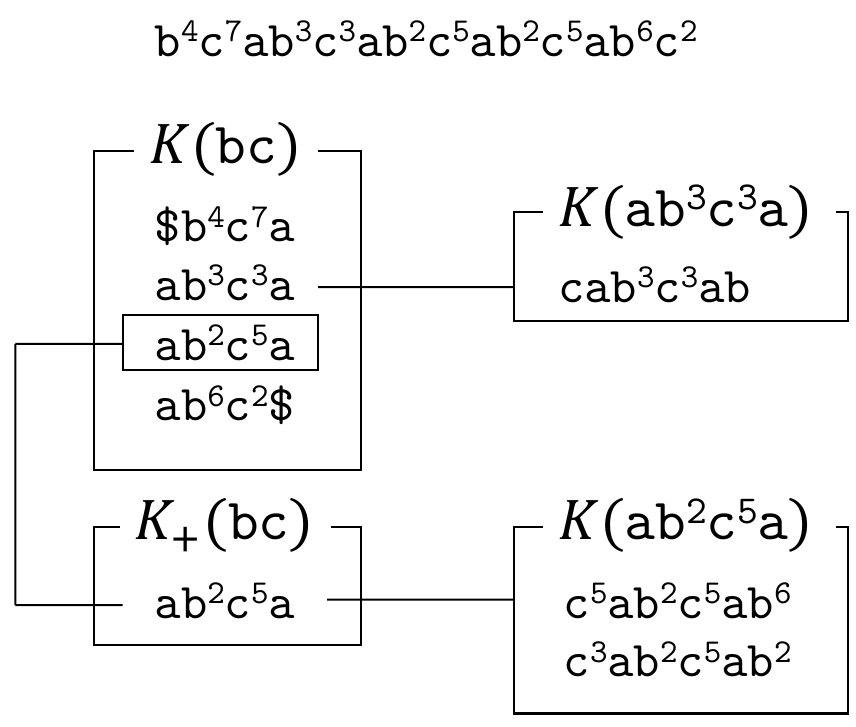}
\caption{
    The bridge $w_1 = \mathtt{ab^2c^5a} \in \Exp(\mathtt{bc})$ is an element of $\Expplus(\mathtt{bc})$ since $|\Exp(w_1)| \geq 2$.
    On the other hand, the bridge $w_2 = \mathtt{ab^3c^3a} \in \Exp(\mathtt{bc})$ is not an element of $\Expplus(\mathtt{bc})$ since $|\Exp(w_2)| < 2$.
}
\label{fig:setK}
\end{center}
\end{figure}
Any MAW $z$ in $\mathcal{M}_4$ is of the form $a\alpha^iu\beta^jb$ with $a, b, \alpha, \beta \in \Sigma, u \in \Sigma^*$, and positive integers $i, j$
where $a, \alpha^i, \beta^j, b$ are the first, the second, the second last, and the last run of $z$, respectively.
By the definition of MAWs, both the suffix $\alpha^iu\beta^jb$ and the prefix $a\alpha^iu\beta^j$ of $z$ occur in $T$.
From this fact, we can obtain the following observations.

\begin{observation}
\label{obs:M4_count}
Each MAW $z \in \mathcal{M}_4$ corresponds to a pair of distinct bridges $(w_1, w_2) \in \Exp(\shrink(z))\times \Exp(\shrink(z))$.
Formally, for each MAW $z = a\alpha^iu\beta^jb \in \mathcal{M}_4$,
there exist characters $a_1, b_1 \in \Sigma \cup \{\$\}$ and integers $i_1 \geq i, j_1 \geq j$
such that $w_1 = a_1 \alpha^{i_1}u\beta^jb, w_2 = a\alpha^iu\beta^{j_1}b_1 \in \Exp(\shrink(z))$
and $w_1 \neq w_2$
(since these two occur in $T$ but $z$ does not occur in $T$).
\end{observation}
This observation gives a main idea of our characterization which is stated in the following lemma.
\begin{lemma} \label{lem:M4toBridges}
For any bridge $w$, $|\{ z \mid \shrink(z) = w, z \in \mathcal{M}_4 \}| \leq |\Exp(w)|(|\Exp(w)|-1)$.
\end{lemma}
\begin{proof}
Let $\mathcal{M}_4(w) = \{ z \mid \shrink(z) = w, z \in \mathcal{M}_4 \}$.
By Observation~\ref{obs:M4_count}, 
each $z \in \mathcal{M}_4(w)$ corresponds to a pair 
$(w_1, w_2) \in \Exp(\shrink(z)) \times \Exp(\shrink(z))$ where $w_1 \neq w_2$.
Let $z_1 = a_1\alpha^{i_1}u\beta^{j_1}b_1, z_2 = a_2\alpha^{i_2}u\beta^{j_2}b_2$ be distinct MAWs in $\mathcal{M}_4(w)$ where $\shrink(z_1) = \shrink(z_2) = w$.
Assume towards a contradiction that $z_1$ and $z_2$ correspond to $(a' \alpha^{i'}u\beta^jb, a\alpha^iu\beta^{j'}b') \in \Exp(w) \times \Exp(w)$.
This implies that, by Observation~\ref{obs:M4_count}, $i = i_1 = i_2, j = j_1 = j_2, a = a_1 = a_2, b = b_1 = b_2$.
Thus $z_1 = z_2$ holds, a contradiction.
Hence, for any distinct MAWs $z_1, z_2 \in \mathcal{M}_4(w)$, $z_1$ and $z_2$ correspond to distinct elements of $\Exp(\shrink(z)) \times \Exp(\shrink(z))$.
Since the number of elements $(w_1, w_2)$ in $\Exp(\shrink(z)) \times \Exp(\shrink(z))$ 
such that $w_1 \neq w_2$ is ${|\Exp(w)|}(|\Exp(w)|-1)$,
this lemma holds.
\end{proof}
Since each MAW $z$ corresponds to an element $(w_1, w_2) \in \Exp(\shrink(z)) \times \Exp(\shrink(z))$ such that $w_1 \neq w_2$,
it is enough for the bound to sum up all $|\Exp(w)|^2$ such that $|\Exp(w)| \geq 2$ holds.
Let $\mathcal{W}$ be the set of bridges $w$ such that $|\Exp(w)| \geq 2$ or $w \in B_2 \cup B_3$.
Let $\mathcal{X} = \sum_{w \in \mathcal{W}}|\Exp(w)|$.
For considering such $\Exp(w)$, we also define a subset $\Expplus^t(w)$ of $\Exp^t(w)$ as follows:
For any string (bridge) $w$ and integer $t > 0$,
\tmnote*{modified}{%
\[
    \Expplus^{t}(w) = \{ w' \mid w' \in \Exp^t(w), |\Exp({w'})| \geq 2 \}.
\]
}%
We also use $\Expplus(w)$ to denote $\Expplus^1(w)$.
Figure~\ref{fig:M4-tree} shows an illustration for $\Exp^i(w), \Expplus^{i}(w), \mathcal{W}$, and $\mathcal{X}$.
\begin{figure}[tb]
    \centering
    \includegraphics[scale=0.4]{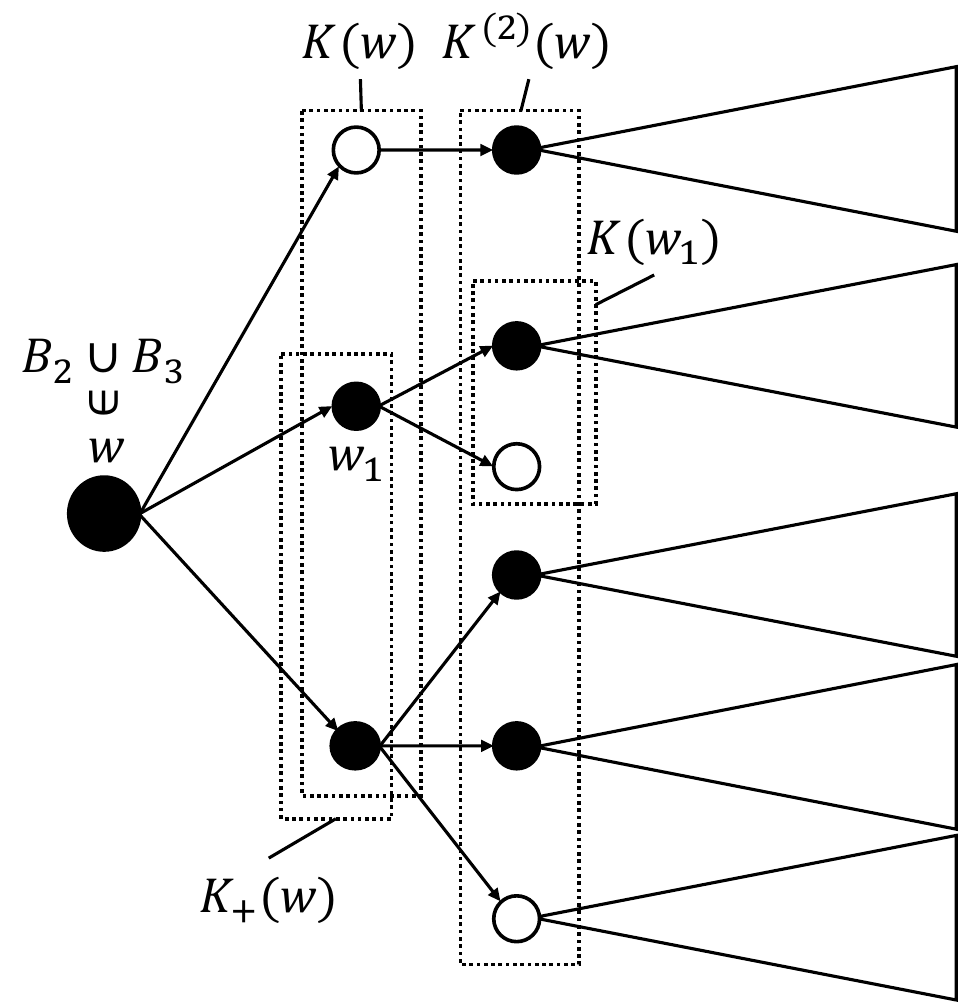}
    \caption{
        This tree shows an illustration for $\Exp^i(w), \Expplus^{i}(w), \mathcal{W}$, and $\mathcal{X}$.
        The root node represents a bridge $w \in B_2 \cup B_3$.
        The set of children of the root corresponds to $\Exp(w)$, 
        namely, each child $x$ represents a bridge such that $\shrink(x) = w$.
        Each black node represents a bridge $x$ such that $|\Exp(x)| \geq 2$ 
        (i.e., each black node has at least two children) or the root.
        Let $W(w)$ be the set of nodes consisting of all the black nodes
        in the tree rooted at a bridge $w \in B_2 \cup B_3$. 
        Then $\mathcal{W}$ is the union of $W(w)$ for all $w \in B_2 \cup B_3$, 
        and $\mathcal{X}$ is the total number of children of black nodes in $\mathcal{W}$.
    }
    \label{fig:M4-tree}
\end{figure}
We give the following lemma that explains relations between $\Exp^i(w), \Expplus^{i}(w)$, and $\mathcal{X}$.
\begin{lemma} \label{lem:boundByX}
\[
    \mathcal{X} = \sum_{w \in B_2 \cup B_3} \left( |\Exp(w)| + \sum_{i=1}^{\lfloor m/2 \rfloor - 1} \sum_{z \in \Expplus^{i}(w)}|\Exp(z)| \right).
\]
\end{lemma}
\begin{proof}
Let $z_{\mathsf{even}}$ be a bridge where $R(z_{\mathsf{even}}) = 2i+2$ for some $i \geq 1$.
Notice that $\shrink(z_{\mathsf{even}}) = c_1c_2 \in B_2$ for some distinct characters $c_1, c_2$.
By the definition of $\Expplus^{i}(\cdot)$,  
if $|\Exp(z_{\mathsf{even}})| \geq 2$, then $z_{\mathsf{even}} \in \Expplus^{i}(c_1c_2)$.
Let $z_{\mathsf{odd}}$ be a bridge where $R(z_{\mathsf{odd}}) = 2i+3$ for some $i \geq 1$.
Notice that $\shrink(z_{\mathsf{odd}}) = c_1c_2^kc_3 \in B_3$ for some characters $c_1, c_2, c_3$ and an integer $k \geq 1$.
By the definition of $\Expplus^{i}(\cdot)$,  
if $|\Exp(z_{\mathsf{odd}})| \geq 2$, then $z_{\mathsf{odd}} \in \Expplus^{i}(c_1c_2^kc_3)$.
Therefore the statement holds.
\end{proof}
This implies that $|\mathcal{M}_4| \leq \sum_{w \in \mathcal{W}} |\Exp(w)|^2 \leq \mathcal{X}^2$.
Thus, if $\mathcal{X} \in O(m)$, $|\mathcal{M}_4| \in O(m^2)$.

\ynnote*{suggestion from reviewer, moved to appendix}{%
We can also observe that $\sum_{i=1}^{\lfloor m/2 \rfloor - 1} \sum_{z \in \Expplus^{i}(w)}|\Exp(z)|$ is the sum of the number of children of black nodes (which have more than a single child) in the tree for $w$.
The number of leaves of the tree is an upper bound for the sum.
It is also clear that $|\Exp(w)|$ can be bounded by the number of leaves of the tree (In Appendix we give a more mathematical description for the above discussion as Observation~\ref{occur_kwsize} and Proposition~\ref{m4_savespace2}).
Consequently, we obtain $|\mathcal{X}| \in O(m)$ as in Lemma~\ref{M4_upper}.
}%

\begin{lemma}
\label{M4_upper}
$|\mathcal{X}| \in O(m)$.
\end{lemma}
\begin{proof}
By Lemma~\ref{lem:boundByX} and the above discussion, we have
\begin{eqnarray*}
\mathcal{X} &=& \sum_{w \in B_2 \cup B_3} \left( |\Exp(w)| + \sum_{i=1}^{\lfloor m/2 \rfloor - 1} \sum_{z \in \Expplus^{i}(w)}|\Exp(z)| \right) \nonumber \\
&\leq& \sum_{w \in B_2 \cup B_3} 2\#w \nonumber \\
&\leq& 2 \left( (m-1) + (m-2) \right) \in O(m).
\end{eqnarray*}
\end{proof}

We are ready to prove Lemma~\ref{M4_number}:
\begin{proof}[Proof of Lemma~\ref{M4_number}]
    $|\mathcal{M}_4| \leq \sum_{w \in \mathcal{W}} |\Exp(w)|^2 \leq |\mathcal{X}|^2
    \leq \left( 2(2m-3) \right)^2 \in O(m^2)$.
\end{proof}

\section{Lower bounds on the number of MAWs for RLE strings}
\label{sec:lower}

In the previous section, we showed a tight bound $|\mathcal{M}_1| = \sigma$, and showed that $|\mathcal{M}_3|$ is unbounded by the RLE size $m$. In this section, we give tight lower bounds for the sizes of $\mathcal{M}_2$, $\mathcal{M}_3$, and $\mathcal{M}_5$ which asymptotically match the upper bounds given in the previous section.
\sinote*{added}{%
Throughout this section, we omit the terminal $\$$ at either end of $T$,
since our lower bound instances do not need them.
}%

\begin{lemma} \label{M2_lowerbound}
There exists a string $T$ such that $| \mathcal{M}_2 | = \sigma'(\sigma'-2)+1$.
\end{lemma}
\begin{proof}
  Let $T = \mathtt{123} \cdots \sigma'$, where all characters in $T$ are mutually distinct. Any bigram occurring in $T$ is of the form $i(i+1)$ with $1 \leq i < \sigma'$.
  Thus, for each $1 \leq i < \sigma'$, bigram $i \cdot j$ with any $j \in \{1, \ldots, i-1, i+2, \ldots, \sigma'\}$ is a type-2 MAW for $T$, and bigram $\sigma' \cdot j$ is a type-2 MAW for $T$. Namely, the set $\mathcal{M}_2$ of type-2 MAWs for $T$ is:
\[
\mathcal{M}_2 = \left\{
\begin{array}{l}
\mathtt{13}, \ldots , \mathtt{1}\sigma', \\
\mathtt{21} , \mathtt{24}, \ldots , \mathtt{2}\sigma', \\
\mathtt{31}, \mathtt{32}, \mathtt{35}, \ldots , \mathtt{3}\sigma', \\
\ldots, \\
(\sigma'-1)\mathtt{1} , \ldots , (\sigma'-1)(\sigma'-2), \\
\sigma' \mathtt{1} , \ldots , \sigma'(\sigma'-1)
\end{array}
\right\}.
\]
Thus we have $| \mathcal{M}_2 | = \sigma'(\sigma'-2)+1$ for this string $T$.
\end{proof}
Since $\sigma' = m$ for the string $T$ of Lemma~\ref{M2_lowerbound}, we obtain a tight lower bound $|\mathcal{M}_2| \in \Omega(m^2)$ in terms of $m$.
The string $T = \mathtt{123} \cdots \sigma'$ can easily be generalized so that $m < n$, where $n = |T|$. For instance, consider $T' = \mathtt{1}^{p_1}{2}^{p_2}{3}^{p_3} \cdots \sigma'^{p_{\sigma'}}$ with $p_i > 1$ for each $i$. The set of type-2 MAWs for $T'$ is equal to that for $T$.

\begin{lemma}
\label{M4_lowerbound}
There exists a string $T$ with $R(T) = m$
such that $| \mathcal{M}_4 | \in  \Omega(m^2)$.
\end{lemma}
\begin{proof}
Consider string
$T = \mathtt{abc}^{p} \cdot \mathtt{ab}^2\mathtt{c}^{p-1} \cdot \mathtt{ab}^3\mathtt{c}^{p-2} \cdot \mathtt{ab}^4\mathtt{c}^{p-3} \cdots \mathtt{ab}^{p-1}\mathtt{c}^2 \cdot \mathtt{ab}^{p}\mathtt{c}\cdot \mathtt{a}$,
where $\mathtt{a}$, $\mathtt{b}$, and $\mathtt{c}$ are mutually distinct characters.
%
\ynnote*{modified based on comments}{%
Then the set of type-4 MAWs for $T$ is a superset of the following set: 
\[
\left\{
\begin{array}{l}
  \mathtt{abca}, \mathtt{abc}^2\mathtt{a}, \ldots, \mathtt{abc}^{p-1}\mathtt{a}, \\
  \mathtt{ab}^{2}\mathtt{ca}, \mathtt{ab}^{2}\mathtt{c}^2\mathtt{a}, \ldots , \mathtt{ab}^2\mathtt{c}^{p-2}\mathtt{a}, \\
  \mathtt{ab}^{3}\mathtt{ca}, \mathtt{ab}^{3}\mathtt{c}^2\mathtt{a}, \ldots , \mathtt{ab}^3\mathtt{c}^{p-3}\mathtt{a}, \\
  \ldots, \\
  \mathtt{ab}^{p-2}\mathtt{ca}, \mathtt{ab}^{p-2}\mathtt{c}^2\mathtt{a}, \\
  \mathtt{ab}^{p-1}\mathtt{c} \mathtt{a} 
\end{array}
\right\}.
\]
Since $m = 3p+1$, we have $|\mathcal{M}_4| > p(p-1)/2 \in \Omega(p^2) = \Omega(m^2)$.
}%
\end{proof}

\begin{lemma} \label{M5_lowerbound}
There exists a string $T$ with $R(T) = m$ such that $| \mathcal{M}_5 | \in  \Omega(m)$.
\end{lemma}
\begin{proof}
Consider string
$T = \mathtt{abc} \cdot \mathtt{ab}^2\mathtt{c}^2 \cdot \mathtt{ab}^3\mathtt{c}^3 \cdots \mathtt{ab}^p\mathtt{c}^p \cdot \mathtt{a}$,
where $\mathtt{a}$, $\mathtt{b}$, and $\mathtt{c}$ are mutually distinct characters.
\ynnote*{modified based on comments}{%
Then the set of type-5 MAWs for $T$ is a superset of the set
\[
  \{\mathtt{b}^{i+1}\mathtt{c}^{i}\mathtt{a} \mid 1 \leq i \leq p-1\}.
\]
Since $m = 3p+1$, $| \mathcal{M}_5 | > p-1 \in \Omega(p) = \Omega(m)$.
}%
\end{proof}

\section{Efficient representations of MAWs for RLE strings}
\label{sec:representation}

Consider a string $T$ that contains $\sigma'$ distinct characters.
In this section, we present compact data structures that can output every MAW for $T$ upon query, using a total of $O(m)$ space, where $m = R(T)$ is the size of $\rle(T)$. We will prove the following theorem:
\begin{theorem} \label{theo:representation}
  There exists a data structure $\mathsf{D}$ of size $O(m)$ which can output all MAWs for string $T$ in $O(|\MAW(T)|)$ time, where $m$ is the RLE-size of $T$.
\end{theorem}

In our representation of MAWs that follows, we store $\rle(T)$ explicitly with $O(m)$ space. 
The following is a general lemma that we can use when we output a MAW from our data structures.

\sinote*{modified}{%
\begin{lemma} \label{lem:maw_in_constant_space}
  For each MAW $w \in \MAW(T)$, $\rle(w)$ of size $R(w)$
  can be retrieved in $O(R(w))$ time from a tuple $(a, i, s, t, b, j)$ and $\rle(T)$,
  where $a,b \in \Sigma$, $0 \leq i,j \leq |T|$, and $0 \leq s,t \leq m$. 
\end{lemma}
\begin{proof}
  When $R(w) = 1$~(i.e. $w \in \mathcal{M}_1$), then since $w$ is of the form $a^i$ with $i \geq 1$, we can simply represent it by $(a, i, 0, 0, 0, 0)$.
  
  When $R(w) \geq 2$, then let $w = aub$.
  When $aub \in \mathcal{M}_2$, then $w = ab$ and thus it can be simply represented by $(a, 1, 0, 0, b, 1)$.
  When $aub \in \mathcal{M}_3 \cup \mathcal{M}_4$, then $a \neq u[1]$ and $b \neq u[|u|]$. Hence it can be represented by $(a, 1, s, t, b, 1)$ where $r_s \cdots r_{t} = \rle(u)$.
  When $aub \in \mathcal{M}_5$, then $a = u[1]$ or $u[|u|] =  b$.
  Let $i,j$ be the maximal integers such that $a^i u' b^j = aub$.
  We can represent it by $(a, i, s, t, b, j)$ with $r_s \cdots r_t = \rle(u')$.
\end{proof}
}%

For ease of discussion, in what follows, we will identify
each MAW $w$ with its corresponding tuple $(a, i, s, t, b, j)$ which takes $O(1)$ space.

\subsection{Representation for $\mathcal{M}_1$}

We have shown that $|\mathcal{M}_1| = \sigma$ (Lemma~\ref{M1_number}), however, $\sigma$ can be larger than $\sigma'$ and $m$. However, a simple representation for $\mathcal{M}_1$ exists, as follows:

\begin{lemma} \label{M1_expression}
There exists a data structure $\mathsf{D}_1$ of $O(\sigma') \subseteq O(m)$ space that can output each MAW in $\mathcal{M}_1$ in $O(1)$ time.
\end{lemma}

\begin{proof}
For ease of explanation, assume that the string $T$ is over the integer alphabet $\Sigma = \{1, \ldots, \sigma\}$ and let $\Sigma' = \{c_1, \ldots, c_{\sigma'}\} \subseteq \{1, \ldots, \sigma\}$.
Let $M = \langle c_1^{p_1}, \ldots, c_{\sigma'}^{p_{\sigma'}} \rangle$ be the list of type-1 MAWs in $\mathcal{M}_1$ that are runs of characters in $\Sigma'$, sorted in the lexicographical order of the characters, i.e. $1 \leq c_1 < \cdots < c_{\sigma'} \leq \sigma$.
We store $M$ explicitly in $O(\sigma')$ space.
When we output each MAW in $\mathcal{M}_1$, we test the numbers (i.e. characters) in $\Sigma = \{1, \ldots, \sigma\}$ incrementally, and scan $M$ in parallel:
For each $c = 1, \ldots, \sigma$ in increasing order, if $c^p \in M$ with some $p > 1$ then we output $c^p$, and otherwise we output $c$.
\end{proof}

\subsection{Representation for $\mathcal{M}_2$}

Recall that $|\mathcal{M}_2| \in O(\sigma'^2) \subseteq O(m^2)$ and this bound is tight in the worst case.
Therefore we cannot store all elements of $\mathcal{M}_2$ explicitly,
as our goal is an $O(m)$-space representation of MAWs.
Nevertheless, the following lemma holds:

\begin{lemma}
\label{M2_expression}
There exists a data structure $\mathsf{D}_2$ of $O(m)$ space that can output each MAW in $\mathcal{M}_2$ in $O(1)$ amortized time.
\end{lemma}
\begin{proof}
If $|\mathcal{M}_2| \in O(m)$, then we explicitly store all elements of $\mathcal{M}_2$.

If $|\mathcal{M}_2| \in \Omega(m)$, then let $\mathsf{D}_2$ be the trie that represents all bigrams that occur in $T$.
See Figure~\ref{fig:data_D2} for a concrete example of $\mathsf{D}_2$.
Note that for any pair $a, b \in \Sigma'$ of \emph{distinct} characters both occurring in $T$, $ab$ is either in $\mathsf{D}_2$ or in $\mathcal{M}_2$.
Since the number of such pairs $a, b$ is $\sigma' (\sigma'-1)$,
we have that $\sigma'^2 = \Theta(|\mathsf{D}_2|+|\mathcal{M}_2|)$, where $|\mathsf{D}_2|$ denotes the size of the trie $\mathsf{D}_2$.
Since $|\mathsf{D}_2| < m$, we have $\sigma'^2 = O(|\mathcal{M}_2|+m)$.
Suppose that the character labels of the out-going edges of each node in $\mathsf{D}_2$ are lexicographically sorted.
When we output each element in $\mathcal{M}_2$,
we test every bigram $ab$ such that $a \neq b$ and $a, b \in \Sigma'$ in the lexicographical order,
and traverse $\mathsf{D}_2$ in parallel in a depth-first manner.
We output $ab$ if it is not in the trie $\mathsf{D}_2$.
This takes $O(\sigma'^2 + |\mathsf{D}_2|) \subseteq O(|\mathcal{M}_2|+m) = O(|\mathcal{M}_2|)$ time, since $|\mathcal{M}_2| \in \Omega(m)$.
\end{proof}

\begin{figure}[t!]
\centering
\includegraphics[width=7cm]{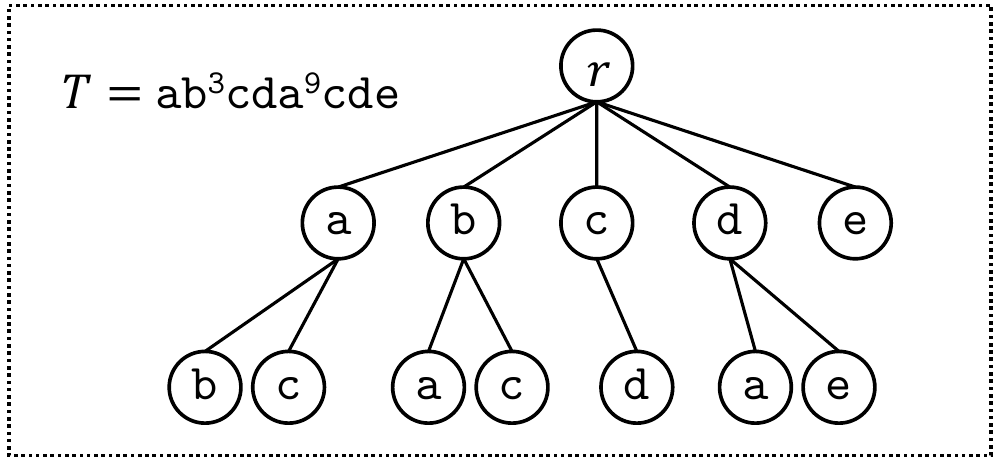}
\caption{
  The trie $\mathsf{D}_2$ for string $T = \mathtt{\$ab}^3\mathtt{cda}^9\mathtt{cde\$}$.
  A bigram $ab$ with $a \neq b$, $a, b \in \Sigma'$ is in $\mathcal{M}_2$ iff $ab$ is not in this trie $\mathsf{D}_2$. For instance, $\mathtt{ae}$ and $\mathtt{db}$ are MAWs of $T$.
}
\label{fig:data_D2}
\end{figure}

\subsection{Representation for $\mathcal{M}_3$}

Recall that the number of MAWs of type 3 in $\mathcal{M}_3$ is unbounded by the RLE size $m$ (Lemma~\ref{M3_number}).
Nevertheless, we show that there exists a compact $O(m)$-space data structure
that can report each MAW in $\mathcal{M}_3$ in $O(1)$ time.

Notice that, by definition, a MAW $aub$ of type 3 is a \bridge\ and therefore, it is of the form $ac^kb$ with $c \in \Sigma'_T \setminus \{a, b\}$ and $k \geq 1$.

We begin with some observations.
For a triple $(a,c,b)$ of characters with $a \neq c$ and $b \neq c$,
let us consider the ordered set $\mathcal{BS}_{acb}(T)$ of \bridge\ substrings of $T$
which are of the form $a c^\ell b$~($\ell \geq 1$),
where the elements in $\mathcal{BS}_{acb}(T)$ are sorted in increasing order of $\ell$.
Let $\ell_{\max} = \max\{\ell \mid a c^\ell b \in \mathcal{BS}_{acb}(T)\}$.
Then, for any $1 \leq k < \ell_{\max}$,
$a c^k b \in \mathcal{M}_3$ iff $a c^k b \notin \mathcal{BS}_{acb}(T)$.
For instance, consider string $T = \mathtt{ac}^3\mathtt{bac}^9\mathtt{b ac}^5 \mathtt{bc}^4\mathtt{e}$ for which $\mathcal{BS}_{\mathtt{acb}}(T) = \{\mathtt{ac}^3\mathtt{b}, \mathtt{ac}^5\mathtt{b}, \mathtt{ac}^9\mathtt{b}\}$.
Then, $\{\mathtt{ac}^1\mathtt{b}, \mathtt{ac}^2\mathtt{b}, \mathtt{ac}^4\mathtt{b}, \mathtt{ac}^6\mathtt{b}, \mathtt{ac}^7\mathtt{b}, \mathtt{ac}^8\mathtt{b}\}$ is the subset of type-3 MAWs of $T$ of the form $\mathtt{a c}^k \mathtt{b}$.
We remark that the above strategy that is based on \bridge\ substrings of the string is not enough to enumerate all elements of $\mathcal{M}_3$,
since e.g. $\mathtt{ac}^3 \mathtt{e}$ and $\mathtt{bc}^2 \mathtt{b}$ are also type-3 MAWs in this running example.
This leads us to define the notion of \emph{combined \bridge s}:
A \bridge\ $a c^\ell b$ is a combined \bridge\ of $T$
if (1) $a c^\ell b$ is not a \bridge\ substring of $T$,
(2) $a c^i b'$ and $a' c^j b$ are \bridge\ substrings of $T$
with $b' \neq b$ and $a' \neq a$,
and (3) $\ell = \min\{i, j\}$.
Let $\mathcal{CB}_c(T)$ denote the set of combined bridges of $T$
with middle character $c$.

\begin{observation} \label{obs:type-3}
  A \bridge\ $ac^kb$ is in $\mathcal{M}_3$ iff
  $ac^kb \notin \mathcal{BS}_{acb}(T)$ and
  either
  (i) $ac^{k'}b \in \mathcal{BS}_{acb}(T)$ with $k' > k$ or
  (ii) $ac^{k'}b \in \mathcal{CB}_{c}(T)$ with $k' \geq k$.
\end{observation}

The type-3 MAWs $\mathtt{ac}^3 \mathtt{e}$ and $\mathtt{bc}^2 \mathtt{b}$
in the running example belong to Case (ii),
since $\mathtt{ac}^3 \mathtt{e}$ is in $\mathcal{CB}_{\mathtt{c}}(T)$
and $\mathtt{bc}^3 \mathtt{b}$ is in $\mathcal{CB}_{\mathtt{c}}(T)$,
respectively.

Observation~\ref{obs:type-3} leads us to the following idea:
For each character $c \in \Sigma'_T$,
let $\mathcal{BS}_c(T) = \bigcup_{a, b \in \Sigma'} \mathcal{BS}_{acb}(T)$ be the ordered set of \bridge\ substrings $z$ of $T$ with $R(z) = 3$ whose middle characters are all $c$.
We suppose that the elements of $\mathcal{BS}_c(T)$ are sorted in increasing order of the exponents $\ell$ of the middle character $c$.
See Figure~\ref{data_D3} for a concrete example for $\mathcal{BS}_c(T)$.

Given $\mathcal{BS}_c(T)$, we can enumerate all type-3 MAWs in $\mathcal{M}_2$ by incrementally constructing a trie $\mathsf{T}_c$ of bigrams.
Initially, $\mathsf{T}_c$ is a trie only with the root.
The algorithm has two stages:
\begin{description}
\item[First Stage:] The first stage deals with Case (i) of Observation~\ref{obs:type-3}. We perform a linear scan over $\mathcal{BS}_c(T)$. When we encounter a \bridge\ substring $ac^\ell b$ from $\mathcal{BS}_c(T)$, we traverse the trie $\mathsf{T}_c$ with the corresponding bigram $ab$.
  \begin{enumerate}
    \item If $ab$ is not in the current trie, then $ac^kb$ for all $1 \leq k < \ell$ are MAWs in $\mathcal{M}_3$.
      After reporting all these MAWs, we create a node $v$ representing $ab$ and store $\ell$.
    \item If $ab$ is already in the current trie, then the value $\hat{\ell}$ stored in the node $v$ which represents $ab$ is less than $\ell$. Then, $ac^k b$ for all $\hat{\ell} < k < \ell$ are MAWs in $\mathcal{M}_3$.
      After reporting all these MAWs, we update the value in $v$ with $\ell$.
  \end{enumerate}
The final trie $\mathcal{T}_c$ after the first stage will be unchanged in the following second stage.

\item[Second Stage:] The second stage deals with Case (ii) of Observation~\ref{obs:type-3}. For each character $a \in \Sigma'_T \setminus \{c\}$, we store the left component $ac^i$ of a \bridge\ substring such that $i$ is the largest exponent of the \bridge\ substrings beginning with $ac$. Let $\mathcal{L}_c$ be the set of $ac^i$'s for all characters $a \in \Sigma'_T \setminus \{c\}$. Similarly, let $\mathcal{R}_c$ be the set of the right components $c^j b$ for all characters $b \in \Sigma'_T \setminus \{c\}$, where $j$ is the largest exponent of the \bridge\ substrings ending with $cb$.
See Figure~\ref{data_D3} for a concrete example for $\mathcal{L}_c$ and $\mathcal{R}_c$.

For each pair of $ac^i \in \mathcal{L}_c$ and $c^j b \in \mathcal{R}_c$, let $ac^\ell b$ be the combined \bridge\ with $\ell = \min\{i, j\}$.
\begin{enumerate}
  \item If $ab$ is not in the trie $\mathsf{T}_c$, then $ac^k b$ for all $1 \leq k \leq \ell$ are MAWs in $\mathcal{M}_3$.
  \item If $ab$ is in the trie $\mathsf{T}_c$, then let $\hat{\ell}$ be the value stored in the node that represents $ab$.
    \begin{enumerate}
    \item If $\hat{\ell} < \ell$, then $ac^kb$ for all $\hat{\ell} < k \leq \ell$ are MAWs in $\mathcal{M}_3$.
    \item If $\hat{\ell} \geq \ell$, then we do nothing.
    \end{enumerate}
\end{enumerate}

\end{description}

\begin{figure}[!t]
\centering
\includegraphics[scale=0.5]{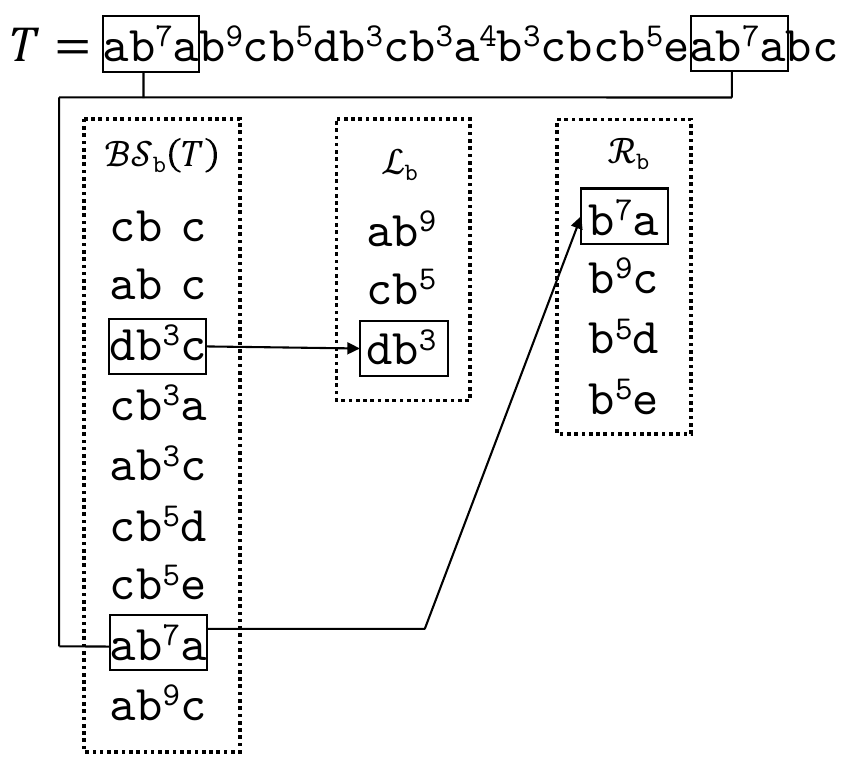}
\caption{$\mathcal{BS}_{\mathtt{b}}$, $\mathcal{L}_{\mathtt{b}}$,
  and $\mathcal{R}_{\mathtt{b}}$ for string $T = \mathtt{ab}^7\mathtt{ab}^9\mathtt{cb}^5\mathtt{db}^3\mathtt{cb}^3\mathtt{a}^4\mathtt{b}^3\mathtt{cbcb}^5\mathtt{eab}^7\mathtt{abc}$ and character $\mathtt{b}$.}
\label{data_D3}
\end{figure}

We have the following lemma:
\begin{lemma} \label{M3_expression}
There exists a data structure $\mathsf{D}_3$ of $O(m)$ space that can output each MAW in $\mathcal{M}_3$ in amortized $O(1)$ time.  
\end{lemma}

\begin{proof}
Analogously to the case of $\mathcal{M}_2$,
if $|\mathcal{M}_2| \in O(m)$, then we can explicitly store all type-3 MAWs in $O(m)$ space.

In what follows, we consider the case where $|\mathcal{M}_2| \in \Omega(m)$.
For each character $c \in \Sigma'_T$,
we perform the above algorithm on $\mathcal{BS}_c(T)$.
The correctness of the algorithm follows from Observation~\ref{obs:type-3}.
Since $\sum_{c \in \Sigma'_T}|\mathcal{BS}_c(T)| \in O(m)$,
the total space requirement of the data structure for all characters in $\Sigma'_T$ is $O(m)$.
Let us consider the time complexity.
The first stage takes $O(m + f) \subseteq O(|\mathcal{M}_3|)$ time,
where $f$ is the number of MAWs reported in the first stage for all characters in $\Sigma'_T$.
The second stage takes $O(|\mathcal{L}_c|\cdot|\mathcal{R}_c|)$ time
for each $c \in \Sigma'_T$.
For each combined \bridge\ $ac^\ell b$ created from $\mathcal{L}_c$ and $\mathcal{R}_c$,
when it falls into Case 1 or Case 2-a, then at least one MAW is reported.
When it falls into Case 2-b, then no MAW is reported.
However, in Case 2-b, there has to be a MAW $a c^k b$ that was reported in the first stage.
Since we test at most one combined \bridge\ for each pair of characters $a,b$,
a MAW $a c^k b$ reported in the first stage is charged at most once.
Therefore, the second stage takes a total of $O(\sum_{c \in \Sigma'_T}|\mathcal{L}_c|\cdot|\mathcal{R}_c|) \subseteq O(|\mathcal{M}_3|)$ time.
%
\end{proof}


\subsection{Representation for $\mathcal{M}_4$}

Recall that $|\mathcal{M}_4| \in O(m^2)$ and this bound is tight in the worst case.
Therefore we cannot store all elements of $\mathcal{M}_4$ explicitly,
as our goal is an $O(m)$-space representation of MAWs.
Nevertheless, the following lemma holds:

\begin{lemma}
\label{lem:M4_expression_sub}
There exists a data structure $\mathsf{D}_4$ of $O(m)$ space that can output each MAW in $\mathcal{M}_4$ in $O(1)$ amortized time.
\end{lemma}

Our data structure $\mathsf{D}_4$ is based on the discussion in Section~\ref{subsec:upperbound_M4}.
We consider the following bipartite graph $G_w = (V_L \cup V_R, E)$ for any \bridge\ $w \in \mathcal{W}$.
We can identify each \bridge\ $a \alpha^i u \beta^j b \in \Exp(w)$ by representing the bridge as a 4-tuple $(a, i, j, b)$.
Let $F_{w}$ be the set of 4-tuples which represents all elements in $\Exp(w)$.
Two disjoint sets $V_L, V_R$ of vertices and set $E$ of edges are defined as follows:
\begin{eqnarray*}
  V_L &=& \{(a, i) \mid \exists (a, i, j, b) \in F_w\}, \\
  V_R &=& \{(j, b) \mid \exists (a, i, j, b) \in F_w\}, \\
  E &=& \{((a, i), (j, b)) \mid \exists (a, i, j, b) \in F_w\}.
\end{eqnarray*}
\ynnote*{added}{%
$V_L$ (resp.~$V_R$) represents the set of the left (resp.~right) parts of bridges in $\mathcal{W}$.
For each edge in $E$ represents a bridge in $\mathcal{W}$.
This implies that $|E| = |\Exp(w)|$.
}%
Assume that all vertices in $V_L$ (resp.~$V_R$) are sorted in non-decreasing order w.r.t. the value $i$ (resp.~$j$) which represents the exponent of corresponding run.
For any $k \in [1, |V_L|]$ and $k' \in [1, |V_R|]$, 
$v_L(k) = (\mathsf{c}_L(k), \mathsf{e}_L(k))$ denotes the $k$-th vertex in $V_L$, and
$v_R(k') = (\mathsf{c}_R(k'), \mathsf{e}_R(k'))$ denotes the $k'$-th vertex in $V_R$.
For any vertex $v_L(k) \in V_L$ and $v_R(k') \in V_R$, we also define 
\[
  E^{LR}_{max}(k) = \max \{\mathsf{e}_R(i) \mid \exists (v_L(k), v_R(i)) \in E\},
\]
\[
  E^{RL}_{max}(k') = \max \{\mathsf{e}_L(i) \mid \exists (v_R(i), v_R(k)) \in E\}.
\]
Figure~\ref{data_D4} gives an illustration for this graph.
\begin{figure}[tb]
\begin{center}
\includegraphics[scale=0.5]{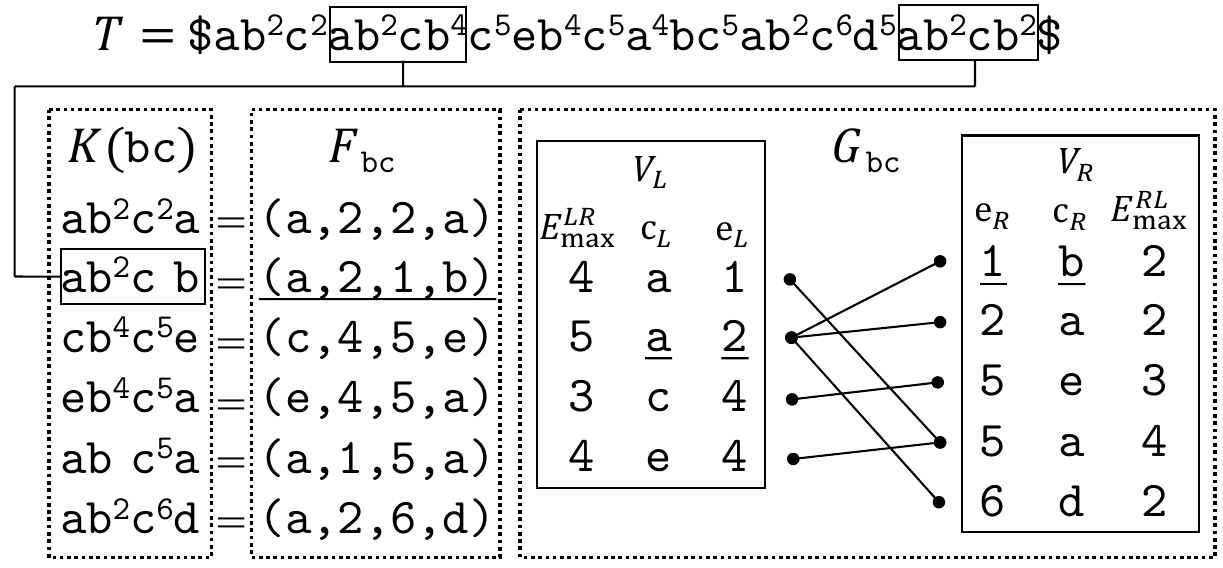}
\caption{
  This figure shows $G_{\texttt{bc}}$ for $T = \mathtt{\$ab^2c^2ab^2cb^4c^5eb^4c^5a^4bc^5ab^2c^6d^5ab^2cb^2\$}$.
  For a bridge $\mathtt{bc}$, $\Exp(\mathtt{bc})$ has 6 bridges.
  $F_{\mathtt{bc}}$ contains 6 tuples which represents all bridges in $\Exp(\mathtt{bc})$.
  For instance, a bridge $\mathtt{ab^2cb} = (\mathtt{a},2,1,\mathtt{b})$ 
  where the first character is $\mathtt{a}$, the exponent of the second run is 2, 
  the exponent of the second last run is 1, and the last character is $\mathtt{b}$.
  $V_L$ is the set of pairs by the left-half of elements in $F_{\mathtt{bc}}$.
  In this example, $V_L$ has 4 vertices $\{(\mathtt{a},1),(\mathtt{a},2),(\mathtt{c},4),(\mathtt{e},4)\}$ which are sorted in non-decreasing order of the second key (representing its exponent).
  $V_R$ is the symmetric set for the right parts.
  Each bridge corresponds to an edge.
  For example, the second bridge $\mathtt{ab^2cb}$ in the figure corresponds to the edge from the second vertex $(\mathtt{a},2)$ in $V_L$ to the first vertex $(1, \mathtt{b})$ in $V_R$.
  Since the number of bridges in $\Exp(\mathtt{bc}) (F_{\mathtt{bc}})$ is 6, the graph has 6 edges.
}
\label{data_D4}
\end{center}
\end{figure}
Due to Observation~\ref{obs:M4_count}, each MAW $z$ of type 4 corresponds to an element of $\Exp(w) \times \Exp(w)$ where $z^{(1)} = w$.
By this idea, we detect each MAW as a pair of vertices in $V_L \times V_R$ which is not an edge in $E$.
The following lemma explains all MAWs which can be represented by the graph.
\begin{lemma}
\label{M4_mawequal}
For any vertices $v_L(k) \in V_L$ and $v_R(k') \in V_R$ of $G_{\alpha u \beta}$,
the string \linebreak
$\mathsf{c}_L(k)\alpha^{\mathsf{e}_L(k)} u \beta^{\mathsf{e}_R(k')}\mathsf{c}_R(k')$ is a MAW iff 
the following three conditions hold (see also Figure~\ref{fig:M4-rep} for an illustration):
\begin{itemize}
  \item $(v_L(k), v_R(k')) \notin E$,
  \item $E^{LR}_{max}(k) \geq \mathsf{e}_R(k')$, and 
  \item $E^{RL}_{max}(k') \geq \mathsf{e}_L(k)$.
\end{itemize}
\end{lemma}
\begin{proof}
If $(v_L(k), v_R(k')) \notin E$, $\mathsf{c}_L(k)\alpha^{\mathsf{e}_L(k)}u \beta^{\mathsf{e}_R(k')}\mathsf{c}_R(k')$ is an absent word.
$E^{LR}_{max}(k) \geq \mathsf{e}_R(k')$ and $E^{RL}_{max}(k') \geq \mathsf{e}_L(k)$ implies that 
$\mathsf{c}_L(k)\alpha^{\mathsf{e}_L(k)}u \beta^{\mathsf{e}_R(k')}$ and $\alpha^{\mathsf{e}_L(k)}u \beta^{\mathsf{e}_R(k')}\mathsf{c}_R(k')$ occur in the string.
Thus $\mathsf{c}_L(k)\alpha^{\mathsf{e}_L(k)}u \beta^{\mathsf{e}_R(k')}\mathsf{c}_R(k')$ is a MAW.

On the other hand, if $(v_L(k), v_R(k')) \in E$, $\mathsf{c}_L(k)\alpha^{\mathsf{e}_L(k)}u \beta^{\mathsf{e}_R(k')}\mathsf{c}_R(k')$ occurs in the text.
$E^{LR}_{max}(k) < \mathsf{e}_R(k')$ implies that $\mathsf{c}_L(k)\alpha^{\mathsf{e}_L(k)}u \beta^{\mathsf{e}_R(k')}$ does not occur in the string.
$E^{RL}_{max}(k') < \mathsf{e}_L(k)$ implies that $\alpha^{\mathsf{e}_L(k)}u \beta^{\mathsf{e}_R(k')}\mathsf{c}_R(k'))$ does not occur in the string.
Thus all three conditions hold if $\mathsf{c}_L(k)\alpha^{\mathsf{e}_L(k)} u \beta^{\mathsf{e}_R(k')}\mathsf{c}_R(k')$ is a MAW.
\end{proof}

\begin{proof}[Proof of Lemma~\ref{lem:M4_expression_sub}]
Let $x$ be the number of outputs.
If $x<m$, we can just store all the MAWs themselves.
Assume that $x \in \Omega(m)$.

For all bridge $w = \alpha u \beta \in \mathcal{W}$, 
$G_{w}$ represents all MAWs which correspond to elements in $\Exp(w) \times \Exp(w)$.
Our data structure $\mathsf{D}_4$ consists of $G_w$ for any $w \in \mathcal{W}$.
It is clear that $G_w$ can be stored in $O(|\Exp(w)|)$ space.
This implies that the size of $\mathsf{D}_4$ is linear in $\mathcal{X}$,
namely, $\mathsf{D}_4$ can be stored in $O(m)$ space (Lemma~\ref{M4_upper}).

We can output all MAWs which are represented by $G_w$ based on Lemma~\ref{M4_mawequal} (see Algorithm~\ref{alg_M4_enumerate}).
For the $k$-th vertex $v_L(k)$, $C$ represents all vertices $v_R(k')$ in $V_B$ 
such that $(v_L(k), v_R(k')) \notin E$ and $E^{RL}_{max}(k') \geq \mathsf{e}_L(k)$ (the first and third condition in Lemma~\ref{M4_mawequal}).
For each vertex in $C$, if $E^{LR}_{max}(k) \geq \mathsf{e}_R(k')$ (the second condition in Lemma~\ref{M4_mawequal}), 
the algorithm outputs a MAW $\mathsf{c}_L(k)\alpha^{\mathsf{e}_L(k)} u \beta^{\mathsf{e}_R(k')}\mathsf{c}_R(k')$.
Then the running time of our algorithm is $O(x + \sum_{w \in \mathcal{W}}|G_w|) \subseteq O(x + m) = O(x)$, since $x \in \Omega(m)$.
\end{proof}

\begin{figure}[tb]
\begin{center}
\includegraphics[scale=0.5]{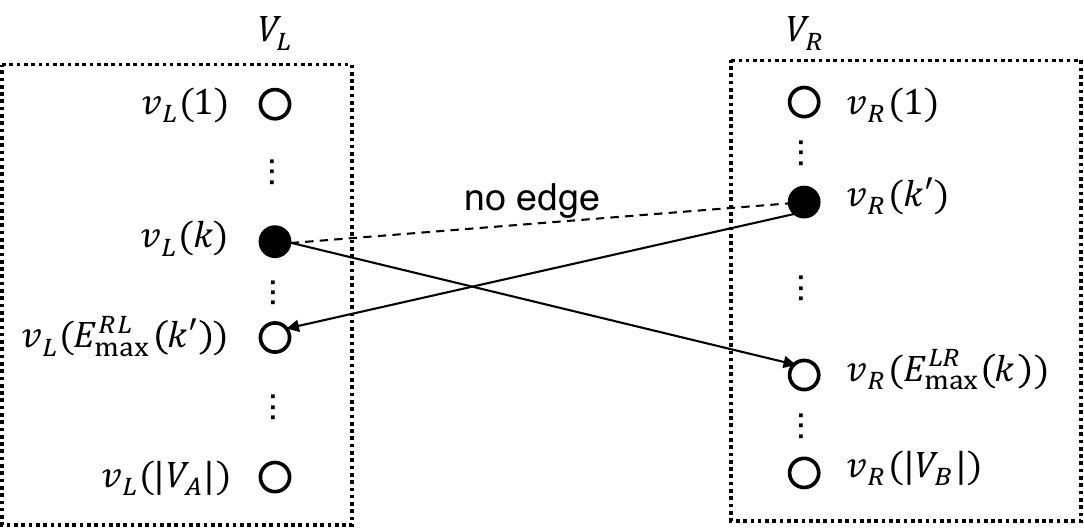}
\caption{
  This is an illustration for Lemma~\ref{M4_mawequal}.
  For the $k$-th vertex $v_L(k) \in V_L$ and $k'$-th vertex $v_R(k') \in V_R$,
  this graph satisfies the three conditions of the lemma.
}
\label{fig:M4-rep}
\end{center}
\end{figure}

\begin{figure}[!t]
\begin{algorithm}[H]
    \caption{Compute all MAWs in $\mathcal{M}_4$}
    \label{alg_M4_enumerate}
    \begin{algorithmic}[1] 
    \renewcommand{\algorithmicrequire}{\textbf{Input:}}
    \renewcommand{\algorithmicensure}{\textbf{Output:}}
    \REQUIRE bipartite graph $G_{\alpha u \beta} = (V_L,V_R,E)$
    \ENSURE all MAWs in $\mathcal{M}_4$ that are associated by $\alpha u \beta$, $a\alpha^{k_1} u \beta^{k_2}b$ for $a,b \in \Sigma, k_1,k_2 \in \mathbb{N}$
    \STATE $C_R \leftarrow V_R$
    \FOR{ each $v_L(k) \in V_L$ }
    \STATE $C = \{ v_R(k') \in C_R \mid \mathsf{e}_R(k') \leq E^{LR}_{max}(k) \} \setminus \{ v \mid (v_L(k), v) \in E \}  $
    \FOR{ each $v_R(k') \in C$ }
    \IF{$E^{RL}_{max}(k') \geq \mathsf{e}_L(v_L(k))$} 
    \STATE output $\mathsf{c}_L(k)\alpha^{\mathsf{e}_L(k)} u \beta^{\mathsf{e}_R(k')}\mathsf{c}_R(k')$
    \ELSE
    \STATE $C_R \leftarrow C_R \setminus \{ v_R(k') \}$
    \ENDIF
    \ENDFOR
    \ENDFOR
    \end{algorithmic}
\end{algorithm}
\end{figure}

\subsection{Representation for $\mathcal{M}_5$}

\begin{lemma}
\label{M5_expression}
There exists a data structure of size $O(m)$ that outputs each element of $\mathcal{M}_5$ in $O(1)$ time.
\end{lemma}
\begin{proof}
  By Lemma~\ref{M5_number}, $|\mathcal{M}_5| \in O(m)$.
  Recall that an element of $M_5$ can be as long as $O(n)$.
  However, using Lemma~\ref{lem:maw_in_constant_space}
  we can represent and store all elements in $\mathcal{M}_5$
  in a total of $O(m)$ space.
  It is trivial that each stored element can be output in $O(1)$ time.
\end{proof}

\section{Conclusions and open questions}

Minimal absent words (MAWs) are combinatorial string objects that can be used in applications such as data compression (anti-dictionaries) and bioinformatics.
In this paper, we considered MAWs for a string $T$ that is described by its run-length encoding (RLE) $\rle(T)$ of size $m$.
We first analyzed the number of MAWs for a string $T$ in terms of its RLE size $m$, by dividing the set $\MAW(T)$ of all MAWs for $T$ into five disjoint types.
Albeit the number of MAWs of some types is superlinear in $m$,
we devised a compact $O(m)$-space representation for $\MAW(T)$
that can output all MAWs in output-sensitive $O(|\MAW(T)|)$ time.

We would like to remark that our $O(m)$-space representation
can be built in $O(m \log m)$ time with $O(m)$ space,
with the help of the \emph{truncated RLE suffix array} (\emph{tRLESA}) data structure~\cite{TamakoshiGIBT15}.
A suffix $s$ of $T$ is called a tRLE suffix of $T$ if
$s = a r_i \cdots r_m$ where the first $a$ is the last character in the previous run $r_{i-1}$.
$\tRLESA(T)$ for $\rle(T) = r_1 \cdots r_m$ is an integer array of length $m$ such that $\tRLESA(T)[i] = k$ iff $a r_i \cdots r_m$ is the $k$-th lexicographically smallest tRLE suffix for $T$.
$\tRLESA$ occupies $O(m)$ space, and can be built in $O(m \log m)$ time with $O(m)$ working space~\cite{TamakoshiGIBT15}.
The details for our tRLESA-based construction algorithm for our $O(m)$-space MAW representation will appear in the full version of this paper.

An interesting open question is whether there exist other compressed representations of MAWs, based on e.g. grammar-based compression~\cite{KiefferY00}, Lempel-Ziv 77~\cite{LZ77}, and run-length Burrows-Wheeler transform~\cite{MakinenN05}.

\section*{Acknowledgments}
This work was supported by JSPS KAKENHI Grant Numbers JP20J11983 (TM), JP18K18002 (YN), JP21K17705 (YN), and by JST PRESTO Grant Number JPMJPR1922 (SI).

We thank the anonymous referees for their comments.

\clearpage

\appendix

\section{Appendix}

We give a supplemental proposition that can be useful
for analyzing the upper bound on the number of MAWs of type 4.

We begin with the following observation:
\begin{observation}
    \label{occur_kwsize}
    For any bridge substring $w \in \Sigma^*$ of $T$,
    \[
        |\Exp(w)| = \#w - \sum_{z \in \Exp(w)} \left( \#z - 1 \right) \leq \#w + |\Expplus(w)| - \sum_{z \in \Expplus(w)}\#z.
    \]
    \end{observation}

Note that $\sum_{z \in \Expplus(w)} \left( \#z - 1 \right) \leq \sum_{z \in \Exp(w)} \left( \#z - 1 \right)$ 
    since $\#z - 1 = 0$ when $z \in \Exp(w)\setminus \Expplus(w)$.
    Below we present Proposition~\ref{m4_savespace2} which gives an upper bound for $\mathcal{X}$.

    \begin{proposition}
    \label{m4_savespace2}
    For any bridge $w$ and $t \geq 1$ such that $|\Exp(w)| \geq 2$,
    \begin{eqnarray}
    |\Exp(w)| + \sum_{i=1}^{t} \sum_{z \in \Expplus^{i}(w)}|\Exp(z)| &\leq& \#w + \sum_{i=1}^{t} |\Expplus^{i}(w)|.
    \end{eqnarray}
    \end{proposition}
    \begin{proof}
    We prove this lemma by induction on $t$.
    By Observation~\ref{occur_kwsize} and $|\Exp(w)| \leq \#w$ for any $w$, we have
    \tmnote*{changed}{%
    \[
    |\Exp(w)| + \sum_{z \in \Expplus(w)}|\Exp(z)| \leq
    (\#w + |\Expplus(w)| - \sum_{z \in \Expplus(w)}\#z)  + \sum_{z \in \Expplus(w)}\#z = \#w + |\Expplus(w)|.
    \]
    }%
    Thus, the statement holds for $t=1$.
    Suppose that the statement holds for some $t' \geq 1$.
    \begin{eqnarray*}
    && |\Exp(w)| + \sum_{i=1}^{t'+1} \sum_{z \in \Expplus^{i}(w)}|\Exp(z)| \\
    &=& |\Exp(w)| + \sum_{w' \in \Expplus(w)} \left( |\Exp(w')| + \sum_{i=1}^{t'} \sum_{z \in \Expplus^{i}(w')} |\Exp(z)| \right) \\
    &\leq& |\Exp(w)| + \sum_{w' \in \Expplus(w)} \left( \#w' + \sum_{i=1}^{t'} |\Expplus^{i}(w')| \right)~\mbox{(by induction hypothesis)}\\
      &\leq& \left( \#w + |\Expplus(w)| - \sum_{w' \in \Expplus(w)}\#w' \right) +  \sum_{w' \in \Expplus(w)} \#w' + \sum_{w' \in \Expplus(w)} \sum_{i=1}^{t'} |\Expplus^{i}(w')| \\
    & & \mbox{(by Observation~\ref{occur_kwsize})}\\
    &=& \#w + |\Expplus(w)| + \sum_{w' \in \Expplus(w)} \sum_{i=1}^{t'} |\Expplus^{i}(w')| \\
    &\leq& \#w + |\Expplus(w)| + \sum_{i=2}^{t'+1} |\Expplus^{i}(w)| \\
    &=& \#w + \sum_{i=1}^{t'+1} |\Expplus^{i}(w)|
    \end{eqnarray*}
    Thus, the statement holds for $t'+1$.
    Therefore, the statement holds for any $t \geq 1$.
    \end{proof}

\bibliographystyle{abbrv}
\bibliography{ref}

\begin{thebibliography}{10}

\bibitem{abs-2105-08496}
T.~Akagi, Y.~Kuhara, T.~Mieno, Y.~Nakashima, S.~Inenaga, H.~Bannai, and
  M.~Takeda.
\newblock Combinatorics of minimal absent words for a sliding window.
\newblock {\em CoRR}, abs/2105.08496, 2021.

\bibitem{Almirantis2017MolecularBiology}
Y.~Almirantis, P.~Charalampopoulos, J.~Gao, C.~S. Iliopoulos, M.~Mohamed, S.~P.
  Pissis, and D.~Polychronopoulos.
\newblock On avoided words, absent words, and their application to biological
  sequence analysis.
\newblock {\em Algorithms for Molecular Biology}, 12(1):5, 2017.

\bibitem{AyadBFHP21}
L.~A.~K. Ayad, G.~Badkobeh, G.~Fici, A.~H{\'{e}}liou, and S.~P. Pissis.
\newblock Constructing antidictionaries of long texts in output-sensitive
  space.
\newblock {\em Theory Comput. Syst.}, 65(5):777--797, 2021.

\bibitem{Barton2014MAWbySA}
C.~Barton, A.~Heliou, L.~Mouchard, and S.~P. Pissis.
\newblock Linear-time computation of minimal absent words using suffix array.
\newblock {\em BMC Bioinformatics}, 15(1):388, 2014.

\bibitem{Beal1996MAWandSymbolicDynamics}
M.~P. B{\'e}al, F.~Mignosi, and A.~Restivo.
\newblock Minimal forbidden words and symbolic dynamics.
\newblock In {\em STACS 1996}, pages 555--566, 1996.

\bibitem{Belazzougui2013ESA}
D.~Belazzougui, F.~Cunial, J.~K{\"a}rkk{\"a}inen, and V.~M{\"a}kinen.
\newblock Versatile succinct representations of the bidirectional
  {Burrows-Wheeler} transform.
\newblock In {\em ESA 2013}, pages 133--144, 2013.

\bibitem{BlumerBHECS85}
A.~Blumer, J.~Blumer, D.~Haussler, A.~Ehrenfeucht, M.~T. Chen, and J.~I.
  Seiferas.
\newblock The smallest automaton recognizing the subwords of a text.
\newblock {\em Theor. Comput. Sci.}, 40:31--55, 1985.

\bibitem{Chairungsee2012PhylogenyByMAW}
S.~Chairungsee and M.~Crochemore.
\newblock Using minimal absent words to build phylogeny.
\newblock {\em Theor. Comput. Sci.}, 450:109 -- 116, 2012.

\bibitem{Charalampopoulos18}
P.~Charalampopoulos, M.~Crochemore, G.~Fici, R.~Mercas, and S.~P. Pissis.
\newblock Alignment-free sequence comparison using absent words.
\newblock {\em Inf. Comput.}, 262:57--68, 2018.

\bibitem{charalampopoulos2018extended}
P.~Charalampopoulos, M.~Crochemore, and S.~P. Pissis.
\newblock On extended special factors of a word.
\newblock In {\em SPIRE 2018}, pages 131--138. Springer, 2018.

\bibitem{CrawfordB018}
T.~Crawford, G.~Badkobeh, and D.~Lewis.
\newblock Searching page-images of early music scanned with {OMR:} {A} scalable
  solution using minimal absent words.
\newblock In {\em {ISMIR} 2018}, pages 233--239, 2018.

\bibitem{CrochemoreHKMPR20}
M.~Crochemore, A.~H{\'{e}}liou, G.~Kucherov, L.~Mouchard, S.~P. Pissis, and
  Y.~Ramusat.
\newblock Absent words in a sliding window with applications.
\newblock {\em Information and Computation}, 270:104461, 2020.

\bibitem{Crochemore1998MAWdefinition}
M.~Crochemore, F.~Mignosi, and A.~Restivo.
\newblock Automata and forbidden words.
\newblock {\em Information Processing Letters}, 67(3):111--117, 1998.

\bibitem{Crochemore2000DCA}
M.~{Crochemore}, F.~{Mignosi}, A.~{Restivo}, and S.~{Salemi}.
\newblock Data compression using antidictionaries.
\newblock {\em Proc. IEEE}, 88(11):1756--1768, 2000.

\bibitem{crochemore2002improved}
M.~Crochemore and G.~Navarro.
\newblock Improved antidictionary based compression.
\newblock In {\em 12th International Conference of the Chilean Computer Science
  Society, 2002. Proceedings.}, pages 7--13. IEEE, 2002.

\bibitem{Fici2006MAWapplication}
G.~Fici.
\newblock {\em Minimal forbidden words and applications}.
\newblock PhD thesis, Universit\`{a} di Palermo and Universit\'{e} Paris-Est
  Marne-la-Vall\'{e}e, 2006.

\bibitem{FiciG19}
G.~Fici and P.~Gawrychowski.
\newblock Minimal absent words in rooted and unrooted trees.
\newblock In {\em {SPIRE} 2019}, pages 152--161, 2019.

\bibitem{FiciRR19}
G.~Fici, A.~Restivo, and L.~Rizzo.
\newblock Minimal forbidden factors of circular words.
\newblock {\em Theor. Comput. Sci.}, 792:144--153, 2019.

\bibitem{Fujishige2016DAWG}
Y.~Fujishige, Y.~Tsujimaru, S.~Inenaga, H.~Bannai, and M.~Takeda.
\newblock Computing {DAWGs} and minimal absent words in linear time for integer
  alphabets.
\newblock In {\em MFCS 2016}, volume~58, pages 38:1--38:14, 2016.

\bibitem{KiefferY00}
J.~C. Kieffer and E.~Yang.
\newblock Grammar-based codes: {A} new class of universal lossless source
  codes.
\newblock {\em {IEEE} Trans. Inf. Theory}, 46(3):737--754, 2000.

\bibitem{koulouras2021significant}
G.~Koulouras and M.~C. Frith.
\newblock Significant non-existence of sequences in genomes and proteomes.
\newblock {\em Nucleic acids research}, 49(6):3139--3155, 2021.

\bibitem{MakinenN05}
V.~M{\"{a}}kinen and G.~Navarro.
\newblock Succinct suffix arrays based on run-length encoding.
\newblock {\em Nord. J. Comput.}, 12(1):40--66, 2005.

\bibitem{MienoKAFNIBT20}
T.~Mieno, Y.~Kuhara, T.~Akagi, Y.~Fujishige, Y.~Nakashima, S.~Inenaga,
  H.~Bannai, and M.~Takeda.
\newblock Minimal unique substrings and minimal absent words in a sliding
  window.
\newblock In {\em {SOFSEM} 2020}, volume 12011 of {\em Lecture Notes in
  Computer Science}, pages 148--160. Springer, 2020.

\bibitem{pratas2020persistent}
D.~Pratas and J.~M. Silva.
\newblock Persistent minimal sequences of {SARS-CoV-2}.
\newblock {\em Bioinformatics}, 36(21):5129--5132, 2020.

\bibitem{TamakoshiGIBT15}
Y.~Tamakoshi, K.~Goto, S.~Inenaga, H.~Bannai, and M.~Takeda.
\newblock An opportunistic text indexing structure based on run length
  encoding.
\newblock In {\em {CIAC} 2015}, volume 9079 of {\em Lecture Notes in Computer
  Science}, pages 390--402. Springer, 2015.

\bibitem{LZ77}
J.~Ziv and A.~Lempel.
\newblock A universal algorithm for sequential data compression.
\newblock {\em IEEE Trans. Inf. Theory}, IT-23(3):337--349, 1977.

\end{thebibliography}

\end{document}